\documentclass[runningheads,a4paper]{llncs}

\makeatother
\makeatletter

\usepackage{lineno,footmisc,marvosym}
\usepackage{fancyhdr}

\usepackage{amssymb}
\usepackage{graphicx}

\usepackage{url}
\urldef{\mailsa}\path|{T.Melissourgos, P.Spirakis}@liverpool.ac.uk|    
\newcommand{\keywords}[1]{\par\addvspace\baselineskip
\noindent\keywordname\enspace\ignorespaces#1}

\let\doendproof\endproof
\renewcommand\endproof{~\hfill\qed\doendproof}
\usepackage{sgame}
\usepackage{graphicx}
\usepackage{tcolorbox}
\usepackage{amsmath}
\usepackage{amsfonts}
\usepackage{verbatim}
\usepackage{bbm}
\usepackage{amssymb}

\usepackage{mathtools}
\usepackage{cases}
\usepackage{cite}
\usepackage{todonotes}
\usepackage{float}
\usepackage{algorithm}
\usepackage[noend]{algpseudocode}
\usepackage{color,soul}

\begin{document}

\mainmatter  

\title{Existence of Evolutionarily Stable Strategies Remains Hard to Decide for a Wide Range of Payoff Values}

\titlerunning{Existence of ESS remains hard to decide for a wide range of payoff values}

%
%
\author{Themistoklis Melissourgos%
\and Paul Spirakis}
%

\institute{Department of Computer Science, University of Liverpool,\\
Ashton Street, Liverpool L69 3BX, United Kingdom\\
\mailsa
}

%
%

%
\maketitle

\begin{abstract}
The concept of an \emph{evolutionarily stable strategy} (ESS), introduced by Smith and Price~\cite{SP}, is a refinement of Nash equilibrium in 2-player symmetric games in order to explain counter-intuitive natural phenomena, whose existence is not guaranteed in every game. The problem of deciding whether a game possesses an ESS has been shown to be \textbf{$\Sigma_{2}^{P}$}-complete by Conitzer~\cite{VC} using the preceding important work by Etessami and Lochbihler~\cite{EL}. The latter, among other results, proved that deciding the existence of ESS is both \textbf{NP}-hard and \textbf{coNP}-hard. In this paper we introduce a \emph{reduction robustness} notion and we show that deciding the existence of an ESS remains \textbf{coNP}-hard for a wide range of games even if we arbitrarily perturb within some intervals the payoff values of the game under consideration. In contrast, ESS exist almost surely for large games with random and independent payoffs chosen from the same distribution~\cite{HRW}.
\end{abstract}

\section{Introduction}

%
%

\subsection{Concepts of Evolutionary Games and Stable Strategies}

Evolutionary game theory has proven itself to be invaluable when it comes to analysing complex natural phenomena. A first attempt to apply game theoretic tools to evolution was made by Lewontin~\cite{RCL} who saw the evolution of genetic mechanisms as a game played between a species and nature. He argued that a species would adopt the ``maximin'' strategy, i.e. the strategy which gives it the best chance of survival if nature does its worst. Subsequently, his ideas where improved by the seminal work of Smith and Price in~\cite{SP} and Smith in~\cite{JMS} where the study of natural selection's processes through game theory was triggered. They proposed a model in order to decide the outcome of groups consisting of living individuals, conflicting in a specific environment. 

The key insight of evolutionary game theory is that a set of behaviours depends on the interaction among multiple individuals in a population, and the prosperity of any one of these individuals depends on that interaction of its own behaviour with that of the others. An \textbf{evolutionarily stable strategy (ESS)} is defined as follows: An infinite population consists of two types of infinite groups with the same set of pure strategies; the \emph{incumbents}, that play the (mixed) strategy $s$ and the \emph{mutants}, that play the (mixed) strategy $t \neq s$. The ratio of mutants over the total population is $\epsilon$. A pair of members of the total population is picked uniformly at random to play a finite symmetric bimatrix game $\Gamma$ with payoff matrix $A_\Gamma$. Strategy $s$ is an ESS if for every $t \neq s$ there exists a constant ratio $\epsilon_t$ of mutants over the total population, such that, if $\epsilon < \epsilon_t$ the expected payoff of an incumbent versus a mutant is strictly greater than the expected payoff of a mutant versus a mutant. For convenience, we say that ``$s$ is an ESS of the game $\Gamma$''.

The concept of ESS tries to capture resistance of a population against invaders. This concept has been studied in two main categories: infinite population groups and finite population groups. The former was the one where this Nash equilibrium refinement was first defined and presented by~\cite{SP}. The latter was studied by Schaffer~\cite{MES} who shows that the finite population case is a generalization of the infinite population one. The current paper deals with the infinite population case which can be mathematically modelled in an easier way and in addition, its results may provide useful insight for the finite population case.

\paragraph{An example.} In order for the reader to conceive the notion of the evolutionarily stable strategy, we give a most explanatory example of the infinite population case. Let us consider a particular species of crab and suppose that each crab's fitness in a specific environment is mainly decided by its capability to find food and use the nutrients from the food in an efficient way. In our crab population a particular mutation makes its appearance, so the crabs born with the mutation grow a significantly larger body size. We can picture the population now, consisting of two distinct kinds of crabs; $\epsilon$ fraction of the population being the large ones and $1-\epsilon$ being the small ones. The large crabs, in fact, have difficulty maintaining the metabolic requirements of their larger body structure, meaning that they need to divert more nutrients from the food they eat and as a consequence, they experience a negative effect on fitness. However, the large crabs have an advantage when it comes to conflicting with the small ones, so they claim an above-average share of the food. To make our framework simple, we will assume that food competition involves pairs of crabs, drawn at random, interacting with each other once, but the reasoning of the analysis is equivalent to interactions that occur (simultaneously or not) between every possible pair, with each individual receiving the mean of the total fitness. When two crabs compete for food, we have the following ``rules'' that apply:  (1) When crabs of the same body size compete, they get equal shares of the food. (2) When a large crab competes with a small crab, the large one gets the majority of the food. (3) In all cases, large crabs experience less of a fitness benefit from a given quantity of food, since some of it is diverted into maintaining their expensive metabolism. (4) When two large crabs compete they experience even less of a fitness benefit as they put considerable effort in fighting. The following bimatrix encloses the rules above in the context of a game.

\begin{center}
\begin{game}{2}{2}[Crab 1][Crab 2]
			& $Small$	& $Large$\\
$Small$   	& $7,7$		& $1,9$\\
$Large$   	& $9,1$		& $4,4$\\
\end{game}
\end{center}

In this setting, we call a given strategy evolutionarily stable if, when the whole population is using this strategy, any small enough group of invaders using a different strategy will eventually die off over multiple generations. This idea is captured in terms of numerical payoffs by saying that, when the entire population is using a strategy $s$, then an arbitrarily small ratio of invaders over the new (blended) population will have strictly lower fitness than the initial population has in the new population. Since fitness translates into reproductive success, and consequently transmitting ones genes to future generations at higher frequencies, strictly lower fitness is assumed by evolutionary principles~\cite{SP} that the reason for a subpopulation (like the users of strategy $t$) to shrink over time through multiple generations and eventually become extinct.

Let us see if any of the two pure strategies is evolutionarily stable. Suppose a population of small crabs gets invaded by a group of large ones (of ratio $\epsilon$ over the whole population). The expected payoff (fitness) of a small crab is:
\begin{align*}
7(1-\epsilon) + 1\epsilon = 7 - 6\epsilon   \qquad &\text{because it meets a small crab with probability}\\
&\text{$1-\epsilon$ and a large one with probability $\epsilon$.}
\end{align*}
The expected payoff of a large crab is:
\begin{align*}
9(1-\epsilon) + 4\epsilon = 9 - 5\epsilon   \qquad &\text{because it meets a small crab with probability}\\ 
&\text{$1-\epsilon$ and a large one with probability $\epsilon$.}
\end{align*}
Clearly, no $\epsilon$ can make the payoff of the small crabs greater than that of the large ones. So, the pure strategy \emph{Small} is not an ESS. Now suppose a population of large crabs gets invaded by a group of small ones (of ratio $\epsilon$ over the whole population). The expected payoff (fitness) of a large crab is:
\begin{align*}
4(1-\epsilon) + 9\epsilon = 4 + 5\epsilon   \qquad &\text{because it meets a large crab with probability}\\
&\text{$1-\epsilon$ and a small one with probability $\epsilon$.}
\end{align*}
The expected payoff of a small crab is:
\begin{align*}
1(1-\epsilon) + 7\epsilon = 1 + 6\epsilon   \qquad &\text{because it meets a large crab with probability}\\
&\text{$1-\epsilon$ and a small one with probability $\epsilon$.}
\end{align*}
In this case, for every $\epsilon \in (0,1)$ the payoff of the large crabs is greater than that of the small ones. So, the pure strategy \emph{Large} is an ESS.

The concept of ESSs can also be extended to mixed strategies. We can think of three natural ways to interpret the notion of probability assignment on the pure strategies of a population. One is, each individual is preprogrammed (through its DNA) to play just a specific pure strategy from a set of strategies and we say that individuals with the same pure strategy are of the same \emph{type}. The group of individuals can be considered to behave as a player with a mixed strategy, defined as a probability vector over the pure strategies used by the group. Each pure strategy's probability equals the ratio of its type's members over the total population (type's \emph{frequency}), because of the simple assumption made, that when two groups conflict one individual from each group is drawn equiprobably to play a bimatrix game. Another one is, each individual is preprogrammed to play a particular mixed strategy. Thus, whoever is drawn will play the specific mixed strategy. The last one is the most general way to think of it, as a blend of the former cases. A group's mixed strategy is defined by its probabilities over the available pure strategies. As soon as one individual is equiprobably picked from each group, the probability over a pure strategy of a group is determined by the sum of the probability each type is picked times the probability this type plays the specific pure strategy. Referring to our previous example, the following three infinite populations of crabs are equivalent: (i) One with $2/3$ of type \emph{Small} and $1/3$ of type \emph{Large}. (ii) One with every crab playing the mixed strategy [$2/3$: \emph{Small}, $1/3$: \emph{Large}]. (iii) One with $1/4$ of type \emph{Small}, $1/4$ playing the mixed strategy [$1/6$: \emph{Small}, $5/6$: \emph{Large}] and $1/2$ playing the mixed strategy [$3/4$: \emph{Small}, $1/4$: \emph{Large}]. Of course in the particular example the individuals cannot have mixed strategies, each one is committed to have a body size for life, but the reasoning holds for other games with strategies that do not exclude each other such as in the Stag-Hunt game. We should mention here, that some games such as Hawk-Dove do not have a pure ESS, but they have a mixed ESS. Other games do not have either.

\subsection{Previous Work}

Searching for the exact complexity of deciding if a bimatrix game possesses an ESS, Etessami and Lochbihler~\cite{EL} invent a nice reduction from the complement of the \textsc{clique} problem to a specific game with an appointed ESS, showing that the \textsc{ess} problem is \textbf{coNP}-hard. They also accomplish a reduction from the \textsc{sat} problem to \textsc{ess}, thus proving that \textsc{ess} is \textbf{NP}-hard too. This makes impossible for the \textsc{ess} to be \textbf{NP}-complete, unless \textbf{NP}=\textbf{coNP}. Furthermore, they provide a proof for the general \textsc{ess} being contained in \pmb{$\Sigma_{2}^{P}$}, the second level of the polynomial-time hierarchy, leaving open the question of what is the complexity class in which the problem is complete.

A further improvement of those results was made by Nisan~\cite{NN}, showing that, given a payoff matrix, the existence of a mixed ESS is \textbf{coDP}-hard. \textbf{DP} is the complexity class, introduced by Papadimitriou and Yannakakis~\cite{PY84}, consisting of all languages $L$ where $L = L_{1} \cap L_{2}$ and $L_{1}$ is in \textbf{NP} and $L_{2}$ is in \textbf{coNP}. Therefore, \textbf{coDP} is the complexity class consisting of all the complement languages of $L$, denoted by $\bar{L}$, where $\bar{L} = \bar{L_{1}} \cup \bar{L_{2}}$ and $\bar{L_{1}}$ is in \textbf{coNP} and $\bar{L_{2}}$ is in \textbf{NP}. Clearly, \textbf{NP} $\subseteq$ \textbf{coDP} , \textbf{coNP} $\subseteq$ \textbf{coDP} and \textbf{coDP} $\subseteq$ \pmb{$\Sigma_{2}^{P}$}. The hardness result is due to a relatively simple reduction from the \textbf{coDP}-complete problem co-\textsc{exact-clique}(for the definition see \cite{PY84}), to \textsc{ess}. A notable consequence of both~\cite{EL} and~\cite{NN} is that the problem of \emph{recognizing} a mixed ESS, once given along with the payoff matrix, is \textbf{coNP}-complete. However, the question of the exact complexity of ESS existence, given the payoff matrix, remained open. A few years later, Conitzer finally settles this question in~\cite{VC}, showing that \textsc{ess} is actually \pmb{$\Sigma_{2}^{P}$}-complete.

On the contrary, Hart et al. \cite{HRW} showed that if the symmetric bimatrix game defined by a $n \times n$ payoff matrix with elements independently randomly chosen according to a distribution $F$ with exponential and faster decreasing tail, such as exponential, normal or uniform, then the probability of having an ESS with just 2 pure strategies in the support tends to 1 as $n$ tends to infinity. In view of this result, and since the basic reduction of \cite{EL} used only 3 payoff values, it is interesting to consider whether ESS existence remains hard for arbitrary payoffs in some intervals.

\subsection{Our Results}

In the reduction of Etessami and Lochbihler that proves \textbf{coNP}-hardness of \textsc{ess} the values of the payoffs used, are $0, \frac{k-1}{k}$ and $1$, for $k \in \mathbb{N}$. A natural question is if the hardness results hold when we \textbf{arbitrarily} perturb the payoff values within respective intervals (in the spirit of smoothed analysis \cite{DSST}). In our work we extend the aforementioned reduction and show that the specific reduction remains valid even after significant changes of the payoff values.

We can easily prove that the evolutionarily stable strategies of a symmetric bimatrix game remain the exact same if we add, subtract or multiply (or do all of them) with a positive value its payoff matrix. However, that kind of value modification forces the entries of the payoff matrix to change in an entirely correlated manner, hence it does not provide an answer to our question. In this work, we prove that if we have partitions of entries of the payoff matrix with the same value for each partition, independent arbitrary perturbations of those values within certain intervals do not affect the validity of our reduction. In other words, we prove that determining ESS existence remains hard even if we perturb the payoff values associated with the reduction. En route we give a definition of ``reduction robustness under arbitrary perturbations'' and show how the reduction under examination adheres to this definition.

In contrast,~\cite{HRW} show that if the payoffs of a symmetric game are random and independently chosen from the same distribution $F$ with ``exponential or faster decreasing tail'' (e.g. exponential, normal or uniform), then an ESS (with support of size 2) exists with probability that tends to 1 when $n$ tends to infinity. 

One could superficially get a non-tight version of our result by saying that (under supposed continuity assumptions in the ESS definition) any small perturbation of the payoff values will not destroy the reduction. However, in such a case (a) the continuity assumptions have to be precisely stated and (b) this does not explain why  the ESS problem becomes easy when the payoffs are random~\cite{HRW}.

In fact, the value of our technique is, firstly, to get as tight as possible ranges of the perturbation that preserve the reduction (and the ESS hardness) without any continuity assumptions, secondly, to indicate the basic difference from random payoff values (which is exactly the notion of partition of payoffs into groups in our definition of robustness, and the allowance of arbitrary perturbation within some interval in each group), and finally, the ranges of the allowed perturbations that we determine are quite tight. For the reduction to be preserved when we independently perturb the values (in each of our partitions arbitrarily), one must show that a system of inequalities has always a feasible solution, and we manage to show this in our final theorem. Our result seems to indicate that existence of an ESS remains hard despite a smoothed analysis \cite{DSST}.

An outline of the paper is as follows: In Section \ref{Robust reductions} we define the robust reduction notion and we provide a first extension of the aforementioned reduction by \cite{EL}. In Section \ref{Extending the reduction} we provide another extended reduction, based on the one from\cite{EL}, that is essentially modified in order to be robust. In Section \ref{Main result} we give our main result and Section \ref{Conclusions} refers to further work and conclusions.

\subsection{Definitions and Notation}

\subsubsection{Background from game theory.}

A \emph{finite two-player strategic form game} $\Gamma = (S_{1}, S_{2}, u_{1}, u_{2})$ is given by finite sets of pure strategies $S_1$ and $S_2$ and utility, or \emph{payoff}, functions $u_1 : S_{1} \times S_{2} \mapsto \mathbb{R}$ and $u_2 : S_{1} \times S_{2} \mapsto \mathbb{R}$ for the row-player and the column-player, respectively. Such a game is called \emph{symmetric} if $S_{1} = S_{2} =: S$ and $u_{1}(i,j) = u_{2}(j,i)$ for all $i,j \in S$. 

In what follows, we are only concerned with finite symmetric two-player strategic form games, so we write $(S, u_1)$ as shorthand for $(S, S, u_{1}, u_{2})$, with $u_{2}(j,i) = u_{1}(i,j)$ for all $i,j \in S$. For simplicity assume $S = {1,...,n}$, i.e., pure strategies are identified with integers $i, 1 \leq i \leq n$. The row-player's \emph{payoff matrix} $A_{\Gamma} = (a_{i,j})$ of $\Gamma = (S, u_1)$ is given by $a_{i,j} = u_{1}(i,j)$ for $i,j \in S$, so $B_{\Gamma}=A_{\Gamma}^{T}$ is the payoff matrix of the column-player. Note that $A_{\Gamma}$ is not necessarily symmetric, even if $\Gamma$ is a symmetric game. 

A \emph{mixed strategy} $s = (s(1),...,s(n))^T$ for $\Gamma = (S, u_1)$ is a vector that defines a probability distribution on $s$ and, in the sequel, we will denote by $s(i)$ the probability assigned by strategy $s$ on the pure strategy $i \in S$. Thus, $s \in X$, where $X = \Big\{s \in \mathbb{R}_{\geq 0}^{n} : \sum_{i=1}^{n}s(i) = 1 \Big\}$ denotes the set of mixed strategies in $\Gamma$, with $\mathbb{R}_{\geq 0}^{n}$ denoting the set of non-negative real number vectors $(x_{1},x_{2},..,x_{n})$. $s$ is called \emph{pure} iff $s(i) = 1$ for some $i \in S$. In that case we identify $s$ with $i$. For brevity, we generally use ``strategy'' to refer to a mixed strategy $s$, and indicate otherwise when the strategy is pure. In our notation, we alternatively view a mixed strategy $s$ as either a vector $(s_{1},...,s_{n})^T$, or as a function $s : S \mapsto \mathbb{R}$, depending on which is more convenient in the context.

The \emph{expected payoff} function, $U_k : X \times X \mapsto \mathbb{R}$ for player $k \in {1,2}$ is given by $U_{k}(s,t) = \sum_{i,j \in S}s(i)t(j)u_{k}(i,j)$, for all $s,t \in X$. Note that $U_{1}(s,t) = s^{T}A_{\Gamma}t$ and  $U_{2}(s,t) = s^{T}A_{\Gamma}^{T}t$. Let $s$ be a strategy for $\Gamma = (S, u_1)$. A strategy $t \in X$ is a \emph{best response} to $s$ if $U_{1}(t,s) = \max_{t' \in X}U_{1}(t',s)$. The \emph{support} supp($s$) of $s$ is the set $\{ i \in S : s(i) > 0 \}$ of pure strategies which are played with non-zero probability. The \emph{extended support} ext-supp($s$) of $s$ is the set $\{ i \in S :U_{1}(i,s) = \max_{x \in X}U_{1}(x,s) \}$ of all pure best responses to $s$.

A pair of strategies $(s,t)$ is a \emph{Nash equilibrium} (NE) for $\Gamma$ if $s$ is a best response to $t$ and $t$ is a best response to $s$. Note that $(s,t)$ is a NE if and only if supp($s$)$\subseteq$ ext-supp($t$) and supp($t$)$\subseteq$ ext-supp($s$). A NE $(s,t)$ is \emph{symmetric} if $s = t$.

\begin{definition}[Symmetric Nash equilibrium]
A strategy profile $(s,s)$ is a symmetric NE for the symmetric bimatrix game $\Gamma = (S, u_1)$ if $s^{T}A_{\Gamma}s \geq t^{T}A_{\Gamma}s$ for every $t \in X$.
\end{definition}

A definition of ESS equivalent to that presented in Subsection 1.1 is:

\begin{definition}[Evolutionarily stable strategy]
A (mixed) strategy $s \in X$ is an evolutionarily stable strategy (ESS) of a two-player symmetric game $\Gamma$ if:
\begin{enumerate}
\item $(s,s)$ is a symmetric NE of $\Gamma$, and
\item if $t \in X$ is any best response to $s$ and $t \neq s$, then $U_{1}(s,t) > U_{1}(t,t)$.
\end{enumerate}
\end{definition}

Due to~\cite{JFN}, we know that every symmetric game has a symmetric Nash equilibrium. The same does not hold for evolutionarily stable strategies (for example ``rock-paper-scissors'' does not have any pure or mixed ESS).

\begin{definition}[{\normalfont\scshape \pmb{ess}} problem]
Given a symmetric two-player normal-form game $\Gamma$, we are asked whether there exists an evolutionarily stable strategy of $\Gamma$.
\end{definition}

\subsubsection{Background from graph theory.}

An \emph{undirected graph} $G$ is an ordered pair $(V,E)$ consisting of a set $V$ of \emph{vertices} and a set $E$, disjoint from $V$, of \emph{edges}, together with an \emph{incidence function} $\psi_G$ that associates with each edge of $G$ an unordered pair of distinct vertices of $G$. If $e$ is an edge and $u$ and $\upsilon$ are vertices such that $\psi_{G}(e) = \{u, \upsilon \}$, then $e$ is said to \emph{join} $u$ and $\upsilon$, and the vertices $u$ and $\upsilon$ are called the \emph{ends} of $e$. We denote the numbers of vertices and edges in $G$ by $\upsilon(G)$ and $e(G)$; these two basic parameters are called the \emph{order} and \emph{size} of $G$, respectively.

\begin{definition}[Adjacency matrix]
The adjacency matrix of the above undirected graph $G$ is the $n \times n$ matrix $A_G := (a_{u \upsilon})$, where $a_{u \upsilon}$ is the number of edges joining vertices $u$ and $\upsilon$ and $n = \upsilon(G)$.
\end{definition}

\begin{definition}[Clique]
A clique of an undirected graph $G$ is a complete subgraph of $G$, i.e. one whose vertices are joined with each other by edges.
\end{definition}

\begin{definition}[{\normalfont\scshape \pmb{clique}} problem]
Given an undirected graph $G$ and a number $k$, we are asked whether there is a clique of size $k$.
\end{definition}

As mentioned earlier, in what follows, $\mathbb{R}_{\geq 0}^{n}$ denotes the set of non-negative real number vectors $(x_{1},x_{2},..,x_{n})$ and $n=|V|$.

\begin{theorem}[Motzkin and Straus~\cite{MS}]\label{MS}
Let $G = (V,E)$ be an undirected graph with maximum clique size $d$. Let $\Delta_1 = \Big \{ x \in \mathbb{R}_{\geq 0}^{n} : \sum_{i=1}^{n} x_{i} = 1 \Big \}$. Then $\max_{x \in \Delta_1} x^{T}A_{G}x = \frac{d-1}{d}$.
\end{theorem}

\begin{corollary}{\label{C_MS}}
Let $G = (V,E)$ be an undirected graph with maximum clique size $d$. Let $A_{G}^{\tau,\rho}$ be a modified adjacency matrix of graph $G$ where its entries with value 0 are replaced by $\tau \in \mathbb{R}$ and its entries with value 1 are replaced by $\rho \in \mathbb{R}$. Let $\Delta_1 = \Big \{ x \in \mathbb{R}_{\geq 0}^{n} : \sum_{i=1}^{n} x_{i} = 1 \Big \}$. Then $\max_{x \in \Delta_1} x^{T}A_{G}^{\tau,\rho}x = \tau + (\rho-\tau)\frac{d-1}{d}$.
\end{corollary}

\begin{proof}
\begin{align*}
x^{T}A_{G}^{\tau,\rho}x & = x^{T}\left[\tau \cdot \mathbf{1} + (\rho-\tau) \cdot A_G\right]x & \text{, where } \mathbf{1} \text{ is the } n \times n \text{ matrix with value 1} \\
& & \text{in every entry.} \\
& = \tau + (\rho-\tau) \cdot x^{T}A_{G}x & \text{, and by Theorem \ref{MS} the result follows.}
\end{align*}
\end{proof}

\begin{corollary}[Etessami and Lochbihler~\cite{EL}]{\label{C_EL}}
Let $G = (V,E)$ be an undirected graph with maximum clique size $d$ and let $l \in \mathbb{R}_{\geq 0}$. Let $\Delta_l = \Big \{ x \in \mathbb{R}_{\geq 0}^{n} : \sum_{i=1}^{n} x_{i} = l \Big \}$. Then $\max_{x \in \Delta_l}x^{T}A_{G}x = \frac{d-1}{d}l^2$.
\end{corollary}

\section{Robust Reductions}\label{Robust reductions}

\begin{definition}[Neighbourhood]  
Let $v \in \mathbb{R}$. An (open) interval $I(v)=[a,b]$ ($I(v)=(a,b)$) with $a<b$ where $a \leq v \leq b$, is called a neighbourhood of $v$ of range $|b-a|$.
\end{definition}

\begin{definition}[Robust reduction under arbitrary perturbations of values]
We are given a valid reduction of a problem to a strategic game that involves a real matrix $A$ of payoffs as entries $a_{ij}$. $A$ consists of $m$ partitions, with each partition's entries having the same value $v(t)$, for $t \in \{1,2,...m\}$. Let $I(v(t)) \neq \emptyset$ be a neighbourhood of $v(t)$ and $w(t) \in I(v(t))$ be an arbitrary value in that neighbourhood. The reduction is called robust under arbitrary perturbations of values if it is valid for all the possible matrices $W$ with entries $w(t)$.
\end{definition}

\subsection{A First Extension of the Reduction from the Complement of the \normalfont\scshape\pmb{clique} \normalfont\pmb{Problem to} \normalfont\scshape\pmb{ess}}\label{Extension}

In the sequel we extend the idea of K. Etessami and A. Lochbihler~\cite{EL} by giving sufficient conditions in order for the reduction to hold. We replace the zeros and ones of their reduction with $\tau>0$ and $\rho<1$ respectively.

Given an undirected graph $G=(V, E)$ we construct the following game $\Gamma_{k,\tau,\rho}(G) = (S, u_1)$ for $\lambda(k) = \frac{k-1}{k}$, where $k \in \mathbb{N}$, and suitable $0<\tau<\rho<1$ to be determined later. Note that from now on we will only consider rational $\tau$ and $\rho$ so that every payoff value of the game is rational.

$S=V \cup \{a,b,c\} $  are the strategies for the players where $a,b,c \notin V$.

$n= |V| $ is the number of nodes.

\begin{itemize}
\item $u_1(i,j)=\rho \text{ for all } i,j \in V \text{ with } (i,j) \in E $ .
\item $u_1(i,j)=\tau \text{ for all } i,j \in V \text{ with } (i,j) \notin E $ .
\item $u_1(z,a)=\rho \text{ for all } z \in S-\{b,c\} $ .
\item $u_1(a,i)=\lambda(k) \text{ for all } i \in V $ .
\item $u_1(y,i)=\rho \text{ for all } y \in \{b,c\} \text{ and } i \in V $ .
\item $u_1(y,a)=\tau \text{ for all } y \in \{b,c\} $ .
\item $u_1(z,y)=\tau \text{ for all } z \in S \text{ and } y \in \{b,c\} $ .
\end{itemize}

Here is an example of the payoff matrix of the strategic game derived from a graph with 3 nodes.

\begin{figure}
\begin{center}
\includegraphics{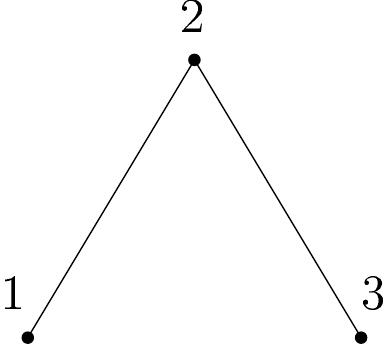}
\caption{The graph $G$.}
\end{center}
\end{figure}

Then the payoff matrix of the row-player is:
\begin{center}
\begin{game}{6}{6}
         & $a$     		& $b$				& $c$				& $1$			& $2$			& $3$\\
$a$   & $\rho$		& $\tau$			& $\tau$			& $(k-1)/k$		& $(k-1)/k$		& $(k-1)/k$\\
$b$   & $\tau$		& $\tau$			& $\tau$			& $\rho$		& $\rho$		& $\rho$\\
$c$   & $\tau$		& $\tau$			& $\tau$			& $\rho$		& $\rho$		& $\rho$\\
$1$   & $\rho$		& $\tau$			& $\tau$			& $\tau$		& $\rho$		& $\tau$\\
$2$   & $\rho$		& $\tau$			& $\tau$			& $\rho$		& $\tau$		& $\rho$\\	
$3$   & $\rho$		& $\tau$			& $\tau$			& $\tau$		& $\rho$		& $\tau$\\
\end{game}
\end{center}

~\\
The transpose of it is the payoff matrix of the column-player. In the sequel we shall use two corollaries of the Motzkin and Strauss theorem, namely, Corollary \ref{C_MS} and Corollary \ref{C_EL}.

\begin{theorem}\label{EL_1}
Let $G=(V, E)$ be an undirected graph. The game $\Gamma_{k,\tau,\rho}(G)$ with $\lambda(k) = \frac{k-1}{k}$ and

\begin{itemize}
\item $ \rho \in \Big(1 - \frac{4}{(n+1)^2},\quad 1 - \frac{1}{(n+1)^2}\Big] \quad \text{and} \quad  \tau \in \Big[(1 - \rho)(n-1),\quad \rho - (1 - \sqrt{1-\rho})^2\Big)  $ 

\text{or}
\item $ \rho \in \Big(1 - \frac{1}{(n+1)^2},\quad 1 \Big) \quad \text{and} \quad  \tau \in \Big[(1 - \rho)(n-1),\quad (1 - \rho)(n-1) + \frac{1}{n+1}\Big) $
\end{itemize}

 has an ESS if and only if G has no clique of size k.
\end{theorem}

\begin{proof}

Let $G=(V,E)$ be an undirected graph with maximum clique size $d$. We consider the game $\Gamma_{k,\tau,\rho}(G)$ above. Suppose $s$ is an ESS of $\Gamma_{k,\tau,\rho}(G)$.

For the reduction we will prove three claims by using contradiction, that taken together show that the only possible ESS $s$ of $\Gamma_{k,\tau,\rho}(G)$ is the pure strategy $a$. Here we should note that these three claims hold not only for the aforementioned intervals of $\tau$ and $\rho$, but for any $\tau,\rho \in \mathbb{R}$ for which $\tau < \rho$.

\begin{claim}[1]
The support of any possible ESS of $\Gamma_{k,\tau,\rho}(G)$ does not contain $b$ or $c$ ($supp(s) \cap \{b,c\} = \emptyset$).
\end{claim}

Suppose $supp(s) \cap \{b,c\} \neq \emptyset$ .

Let $t \neq s$ be a strategy with $t(i)=s(i)$ for $i \in V$, $t(y)=s(b)+s(c)$ and $t(y')=0$ where $y,y' \in \{b,c\}$ such that $y \neq y'$ and $s(y)=min\{s(b),s(c)\}$. Since $u_1(b,z)=u_1(c,z)$ for all $z \in S$,
\begin{align*}
&U_1(t,s) = \sum_{i \in V}t(i)U_1(i,s) + (t(b) + t(c))U_1(b,s) + t(a)U_1(a,s) \;, \\
&U_1(s,s) = \sum_{i \in V}s(i)U_1(i,s) + (s(b) + s(c))U_1(b,s) + s(a)U_1(a,s) \;,
\end{align*}
\\
which yields $U_1(t,s) = U_1(s,s)$ and so $t$ is a best response to $s$. Also,
\begin{align*}
&U_1(s,t) = \sum_{i \in V}s(i)U_1(i,t) + (s(b) + s(c))U_1(b,t) + s(a)U_1(a,t) \;, \\
&U_1(t,t) = \sum_{i \in V}t(i)U_1(i,t) + (t(b) + t(c))U_1(b,t) + t(a)U_1(a,t) \;,
\end{align*}
 which yields $U_1(s,t) = U_1(t,t)$. But this is a contradiction since it should be $U_1(s,t) > U_1(t,t)$ as $s$ is an ESS.

\begin{claim}[2]
The support of any possible ESS of $\Gamma_{k,\tau,\rho}(G)$ contains $a$ ($supp(s) \nsubseteq V$).
\end{claim}

Suppose $supp(s) \subseteq V$ .

Then, we denote by $A_G$ the adjacency matrix of the graph $G$.
\begin{align*}
U_1(s,s) = \sum_{i,j \in V}s(i)s(j)u_1(i,j) &= x^{T}A_{G,\tau,\rho}x \\
&\leq \tau + (\rho - \tau)\frac{d-1}{d} \quad \text{(by Corollary \ref{C_MS})}\\ 
&< \rho = U_1(b,s) \text{\qquad for every $\rho > \tau$ .}
\end{align*}

But this is a contradiction since $s$ is an ESS and therefore a NE. From Claim (1) and Claim (2), it follows that $a \in supp(s)$, i.e. $s(a)>0$ .

\begin{claim}[3]
$s(a)=1$ .
\end{claim}

Suppose $s(a)<1$ .

Since $(s,s)$ is a NE, $a$ is a best response to $s$ and $a \neq s$. Then $U_1(s,a) = \sum_{z \in supp(s)}s(z)u_1(s,a)=\rho=U_1(a,a)$. But this is also a contradiction since it should be $U_1(s,a) > U_1(a,a)$ as $s$ is an ESS. Therefore, the only possible ESS of $\Gamma_{k,\tau,\rho}(G)$ is the pure strategy $a$.

Now we show the following lemma, which concludes also the proof of Theorem \ref{EL_1}.
\begin{lemma}{\label{Lemma1}}
The game $\Gamma_{k,\tau,\rho}(G)$ with the requirements of Theorem \ref{EL_1} has an ESS (strategy $a$) if and only if there is no clique of size $k$ in graph $G$.
\end{lemma} 

\begin{proof}
We consider two cases for $k$:
\subsubsection{Case 1: $d<k$}. Let $t \neq a$ be a best response to $a$. Then $supp(t) \subseteq V \cup \{a\}$ .

Let $r= \sum_{i \in V}t(i)$. So $r>0$($t \neq a$) and $t(a)=1-r$ . Combining Corollary \ref{C_MS} and \ref{C_EL} we get,

\begin{align*}
U_1(t,t) - U_1(a,t) =& \sum_{i,j \in V}t(i)t(j)u_1(i,j) + r\cdot t(a) \cdot \rho + \\
& + t(a) \cdot r \cdot \frac{k-1}{k} + t(a)^2 \cdot \rho - \Big[r \cdot \frac{k-1}{k} + t(a) \cdot \rho \Big] \\ 
\leq & \Big[ \tau + (\rho - \tau) \frac{d-1}{d} \Big]r^2 + r(1-r) \cdot \rho + \\
& + (1-r)r \frac{k-1}{k} + (1-r)^2 \cdot \rho - r \frac{k-1}{k} - (1-r) \cdot \rho \\
= & \Big[ \tau + (\rho - \tau) \frac{d-1}{d} \Big]r^2 - \frac{k-1}{k} r^2 \\
= & \frac{r^2}{d} \Big[ \tau + \rho (d-1) - d \frac{k-1}{k} \Big] \\
= & \frac{r^2}{d}E \qquad \text{, where } E=\tau + \rho (d-1) - d \frac{k-1}{k} \;.
\end{align*}

If we can show that $E<0$ then strategy $a$ is an ESS. We now show why $E<0$:

Let us define the following function,
\begin{align*}
f(k,d,\rho) = d \frac{k-1}{k} - \rho (d-1)  \quad  & \text{, with the restrictions: } k \geq d+1 , 1 \leq d \leq n \\
& \text{ and } \rho \in (0,1) \;.
\end{align*}

Then we define the function $g(d,\rho)$:
\begin{align}
g(d,\rho) = \min_{k}f(k,d,r) &=  d \frac{d}{d+1} - \rho (d-1) = (1 - \rho)(d-1) + \frac{1}{d+1} \;.
\end{align}

By examining the first and second partial derivative with respect to variable $d$, we find the minimum of function $g(d,\rho)$:
\begin{align}
h(\rho) = \min_{d}g(d,\rho) = \rho - (1 - \sqrt{1-\rho})^2  \text{\qquad , for \quad} d^* = \frac{1}{\sqrt{1 - \rho}} - 1 \;.
\end{align}

Now there are two subcases. The maximum clique size may be impossible to reach the value of $d^*$, or it could reach it, depending on the size of $n=|V|$ . 

\paragraph{Subcase i)}  $n < \frac{1}{\sqrt{1 - \rho}} - 1$  or equivalently: $\rho > 1 - \frac{1}{(n+1)^2}$ \;.

From the partial derivatives of function $g(d,\rho)$ with respect to variable $d$ we know that it is a strictly decreasing function for $d<d^*$.
And given that $d \leq n$, from (1) we get: 
\
\begin{align}
h(\rho) = (1 - \rho)(n-1) + \frac{1}{n+1} \text{\qquad , for \quad} 1 - \frac{1}{(n+1)^2} < \rho < 1 \;.
\end{align}

\paragraph{Subcase ii)}  $n \geq \frac{1}{\sqrt{1 - \rho}} - 1$  or equivalently: $\rho \leq 1 - \frac{1}{(n+1)^2}$ \;.

By examining the first and second partial derivative with respect to variable $\rho$, we find the plot of function $h(\rho)$ to be:

\begin{figure}
\begin{center}
\includegraphics[scale=0.25]{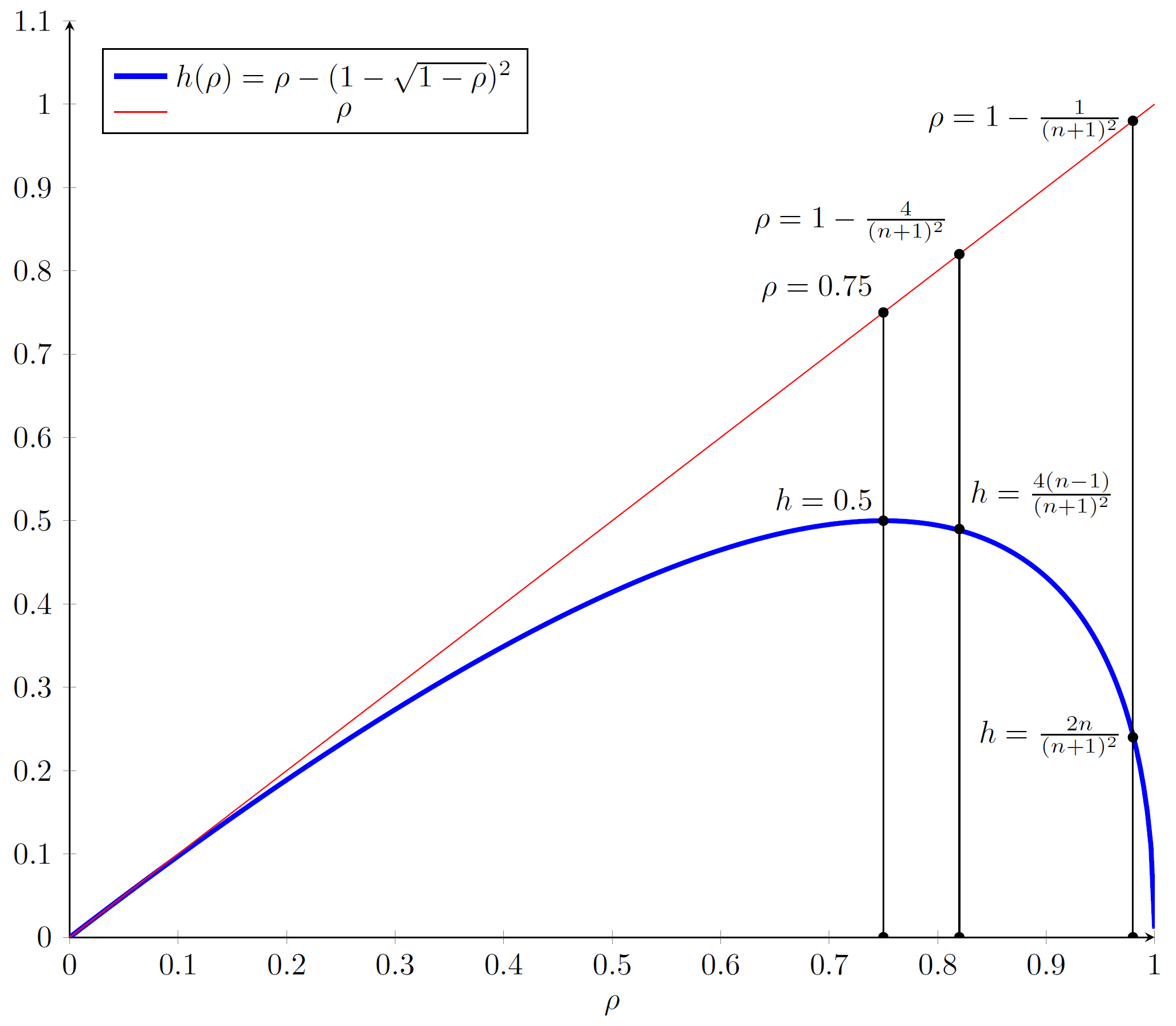}
\end{center}
\end{figure}

%
%
%
%

As we can see, the maximum of $h(\rho)$ is $\frac{1}{2}$ and it is achieved when $\rho = \frac{3}{4}$ .
\\
\\

\paragraph{Interval a)}  $\frac{3}{4} < \rho \leq 1 - \frac{1}{(n+1)^2}$ .

The monotonicity of $h(\rho)$ in this interval implies that its minimum is achieved for $\rho^*=1 - \frac{1}{(n+1)^2}$ . Thus if we want a minimum independent of $\rho$, from (2) we get:
\begin{align}
\min_{\rho}h(\rho) = 1 - \frac{1}{(n+1)^2} - \left(1 - \sqrt{1 - \left(1 - \frac{1}{(n+1)^2}\right)}\right)^2 = \frac{2n}{(n+1)^2} \;.
\end{align}

\paragraph{Interval b)}  $0 < \rho \leq \frac{3}{4}$ .

The monotonicity of $h(\rho)$ in this interval implies that there is no minimum point, but when $\rho$ gets arbitrarily close to zero then $h(\rho)$ goes arbitrarily close to zero as well, i.e. $\lim_{\rho \rightarrow 0^+}h(\rho) = 0$ .
\\
\\
To sum up:
\[
\tau^* = \min_{k,d}f(k,d,\rho) = 
\begin{cases} 
\rho - (1 - \sqrt{1-\rho})^2        &  \text{, if }   0 < \rho \leq 1 - \frac{1}{(n+1)^2}  \text{ , from (2)}\\
\\
(1 - \rho)(n-1) + \frac{1}{n+1}  &  \text{, if }   1 - \frac{1}{(n+1)^2} < \rho < 1     \text{ , from (3)}\\
\end{cases}
\] 

or if we want the minima to be independent of $\rho$ when possible:
\[
\tau^* = \min_{k,d,\rho}f(k,d,\rho) = 
\begin{cases} 
\rho - (1 - \sqrt{1-\rho})^2       &  \text{ , if}  \quad   0 < \rho \leq \frac{3}{4}  \\	
\\
\frac{2n}{(n+1)^2}	&  \text{ , if}  \quad   \frac{3}{4} < \rho \leq 1 - \frac{1}{(n+1)^2}    \text{ , from (4)}\\
\\
\frac{1}{n+1}		&  \text{ , if}  \quad   1 - \frac{1}{(n+1)^2} < \rho < 1       \text{ , from (3).}\\
\end{cases}
\] 
\\
Therefore, depending on the interval that $\rho$ belongs to, we can demand $\tau$ to be strictly less than $\tau^*$ , making  $U_1(t,t) - U_1(a,t)$ negative. We conclude that when $d<k$ then strategy $a$ is an ESS.

\subsubsection{Case 2: $d \geq k$}. Let $C \subseteq V$ be a clique of $G$ of size $k$. Then $t$ with $t(i)=\frac{1}{k}$ for $i \in C$ and $t(j)=0$ for $j \in S \setminus C$ is a best response to $a$ and $t \neq a$, and
\begin{align*}
&U_1(t,t) = \sum_{i,j \in C}t(i)t(j)u_1(i,j)= \frac{1}{k^2} \cdot (k-1)k \cdot \rho + \frac{1}{k^2}k \cdot \tau = \frac{(k-1) \rho + \tau}{k}  \;, \\
&U_1(a,t) = \frac{k-1}{k} \;.
\end{align*}

Then,
\begin{align*}
U_1(t,t) - U_1(a,t) &= \frac{1}{k} \Big[\tau - (1 - \rho)(k-1) \Big] \\
&= \frac{1}{k}E' \qquad \text{, where } E'=\tau - (1-\rho)(k-1) \;.
\end{align*}

If $E' \geq 0$ then $a$ cannot be an ESS. We explain why $E' \geq 0$:

Let's define the following function:
\begin{equation*}
y(k,\rho) =(1 - \rho)(k-1)  \text{\qquad , with the restrictions: } k \leq d  \text{ and } \rho \in (0,1) \;.
\end{equation*}

Then we define the function $z(d,\rho)$:
\begin{align*}
z(d,\rho) = \max_{k}y(k,\rho) &= (1 - \rho)(d-1) \;,
\end{align*}

so,
\begin{align*}
\tau^{**} = \max_{d}z(d,\rho) &= (1 - \rho)(n-1) \;.
\end{align*}

Now, given that $\tau$ needs to be at least $\tau^{**}$ but strictly less than $\tau^{*}$ the following should hold:
\begin{align*}
(1 - \rho)(n-1) < \rho - (1 - \sqrt{1-\rho})^2  \text{ \quad, or equivalently, \quad} \rho > 1 - \frac{4}{(n+1)^2} \;.
\end{align*}

So we conclude that when $d \geq k$ then strategy $a$ is not an ESS. This completes the proof of Lemma \ref{Lemma1} and Theorem \ref{EL_1}.

\end{proof}
\end{proof}

\subsection{An Interesting Consequence}
An interesting consequence of the analysis above, is the fact that if we could possess an algorithm, called \textbf{DecESS}, which decides in polynomial time whether a game $\Gamma_{k,\tau,\rho}(G)$ has an ESS, then the maximum clique size of graph $G$ can also be found in polynomial time using the following binary search algorithm.

\begin{algorithm}[H]
\caption{Binary clique search}\label{alg:bcs}
\begin{algorithmic}[1]
\State $min \gets 1$ \Comment{Initialization}
\State $max \gets n$
\While{$(min \not= max)$} \Comment{We have more than 1 candidates for max clique size}
\State $mid \gets \lceil (min + max)/2 \rceil$ \Comment{Search in the middle of the set}
\If{ DecESS($\Gamma_{mid,\tau',\rho'}(G)$)=''yes''} \Comment{By \pmb{Theorem \ref{EL_1}}...}
\State $max \gets mid-1$ \Comment{...max clique size is less than ''mid''}
\Else \Comment{By \pmb{Theorem \ref{EL_1}}...}
\State $min \gets mid$ \Comment{...max clique size is at least ''mid''}
\EndIf
\EndWhile\label{bcsendwhile}
\State \textbf{return} $min$ \Comment{The clique size is ''min''}
\end{algorithmic}
\end{algorithm}

In this algorithm we supposed there is an algorithm, called \textbf{DecESS}, that uses as input the game: $\Gamma_{mid,\tau',\rho'}(G)$, where the $mid$ value depends on the current $min$ and $max$ values of the algorithm, and $\tau'$ and $\rho'$ are picked in the intervals:
\begin{itemize}
\item $ \rho' \in \Big(1 - \frac{4}{(max+1)^2},\quad 1 - \frac{1}{(max+1)^2}\Big] \quad \text{and} \quad  \tau' \in \Big[(1 - \rho)(max-1),\quad \rho - (1 - \sqrt{1-\rho})^2\Big)  $ \\
\\
\text{or}
\item $ \rho' \in \Big(1 - \frac{1}{(max+1)^2},\quad 1 \Big) \quad \text{and} \quad  \tau' \in \Big[(1 - \rho)(max-1),\quad (1 - \rho)(max-1) + \frac{1}{max+1}\Big) $ ,
\end{itemize}
for the current value of $max$ in each loop of the algorithm. The output of \textbf{DecESS} is: ''yes'', if there exists an ESS in $\Gamma_{mid,\tau',\rho'}(G)$ and ''no'', otherwise. So we construct a new game $\Gamma_{mid,\tau,\rho}(G)$ every time $min$ or $max$ (and therefore $mid$) are changed. Note that while the binary search runs, the maximum possible clique size of the graph ($max$) changes, so, we can modify the intervals of our $\tau'$ and $\rho'$ as if we had a new graph with $|V|=max$ instead of $n$.

As the \textsc{clique} problem has been proved to be \textbf{NP}-complete, to find the maximum clique size of a given graph (\textsc{max-clique} problem) is \textbf{NP}-hard, thus, possession of the above mentioned algorithm would yield that \textbf{P}=\textbf{NP}.

To determine the time complexity of the \textbf{Binary clique search} algorithm let's suppose that the steps which \textbf{DecESS} needs are $R(n) \in O(n^w)$ for some constant $w$. From the algorithm we can derive the recurrent relation for the steps needed:
\begin{align*}
T(m) & = 4 + R(n) + T(\lceil \tfrac{m}{2} \rceil) \\
& = 4 + R(n) + 4 + R(n) + T\left(\left\lceil \tfrac{m}{4} \right\rceil\right) \\
& = ...  \\
\text{(the steps of} & \text{ \textbf{DecESS} are not dependent on the size $m$ of the search list)} \\
\text{and in general }\\
 T(m) & = (4 + R(n))i + T\left(\left\lceil \frac{m}{2^i} \right\rceil\right) \;.
\end{align*}

The base case is:  
\begin{equation*}
T\left(\left\lceil \frac{m}{2^i} \right\rceil\right) = T(1) 
\Rightarrow \left\lceil \frac{m}{2^i} \right\rceil = 1  \Rightarrow i = \lceil \log_2 m \rceil \;,
\end{equation*}

so, the latter equation is:
\begin{equation*}
 T(m) = (4 + R(n)) \lceil \log_2 m \rceil + T(1) \;.
\end{equation*}

Our initial condition is $T(1)=1$ .

Also, $m=n$. (We wrote $m$ instead of $n$ above because the steps $R(n)$ of \textbf{DecESS} do not depend on the size $m$ of the search list of each \textbf{Binary clique search}'s loop, they only depend on the number of $G$'s vertices $n=|V|$.) So, if we count the steps for the initialization of the variables along with the return command, the steps needed by \textbf{Binary clique search} are:
\begin{equation*}
 T(n) = (4 + R(n)) \lceil \log_2 n \rceil + 4 \;,
\end{equation*}

which yields:
\begin{equation*}
 T(n) \in O(n^{w}\log_2 n) \;.
\end{equation*}

To sum up, if we have a polynomial time algorithm \textbf{DecESS} which decides if the game $\Gamma_{k,\tau,\rho}(G)$ has an ESS, then the \textsc{max-clique} problem is solvable in polynomial time, as we can always reduce an undirected graph $G=(V, E)$ to $\lceil \log_2 n \rceil$ number of $\Gamma_{k,\tau,\rho}(G)$ games, each of them in polynomial time and eventually find the maximum clique size of $G$ in polynomial time using \textbf{Binary clique search}.

All in all, supposing the reduction from the graph to the game requires $O(n^{r})$ time for some constant $r$ (as shown by~\cite{EL}), then, the assumption of a \textbf{DecESS} in \textbf{P} yields that the \textsc{max-clique} problem requires $O(n^{r+w}\log_2 n)$ time and \textbf{P}=\textbf{NP}.

\begin{corollary}
The \textsc{ess} problem with payoff values in the domains given in Theorem \ref{EL_1} is \textbf{coNP}-hard.
\end{corollary}

\section{Extending the Reduction with Respect to $\lambda(k)$}\label{Extending the reduction}
We now prove a generalization of the latter reduction for $\lambda(k) = 1- \frac{1}{k^x}$, with $x \geq 3$:

\begin{theorem}\label{EL_x}
Let $G=(V, E)$ be an undirected graph. The game $\Gamma_{k,\tau,\rho}^{x}(G)$ with $\lambda(k) = 1- \frac{1}{k^x}$, for $x \geq 3$ and 
\begin{itemize}
\item $ \rho \in \left(1 + \frac{n^{x-1}-2^x}{2^{x}n^{x-1}(n-1)}, \quad 1 + \frac{(n+1)^x - n2^x}{2^{x}(n+1)^{x}(n-1)} \right] \quad \text{and} \\ 
	\tau \in \left[(1 - \rho)(n-1) + 1 - \frac{1}{n^{x-1}},\quad 1 - \frac{1}{2^x}\right)  $ 

\text{or}
\item $ \rho \in \left(1 + \frac{(n+1)^x - n2^x}{2^{x}(n+1)^{x}(n-1)},\quad + \infty \right) \quad \text{and} \\
	\tau \in \left[(1 - \rho)(n-1) + 1 - \frac{1}{n^{x-1}},\quad (1-\rho)(n-1) + 1 - \frac{n}{(n+1)^x}\right) $
\end{itemize}
has an ESS if and only if G has no clique of size k.
\end{theorem}

\begin{proof}

Let $G=(V,E)$ be an undirected graph with maximum clique size $d$. We consider the game $\Gamma_{k,\tau,\rho}(G)$ defined in Subsection \ref{Extension}, with the only difference that now, we substitute payoffs of value $\frac{k-1}{k}$ with new payoffs $\frac{k^x-1}{k^x}$, meaning we make the change $k \leftarrow k^x$. Suppose $s$ is an ESS of $\Gamma_{k,\tau,\rho}^{x}(G)$.

In this case, the same analysis as in Subsection \ref{Extension} is similarly applied up to the point where we prove that the only possible ESS of $\Gamma_{k,\tau,\rho}^{x}(G)$ is the pure strategy $a$.
Now we proceed to show the following lemma, which concludes also the proof of Theorem \ref{EL_x}.
\begin{lemma}{\label{Lemma2}}
The game $\Gamma_{k,\tau,\rho}^{x}(G)$ with the requirements of Theorem \ref{EL_x} has an ESS (strategy $a$) if and only if there is no clique of size $k$ in graph $G$. 
\end{lemma}

\begin{proof}
We consider again two cases for $k$:

\subsubsection{Case 1: $d<k$}. Let $t \neq a$ be a best response to $a$. Then $supp(t) \subseteq V \cup \{a\}$.

Let $r= \sum_{i \in V}t(i)$. So $r>0, (t \neq a$) and $t(a)=1-r$. Combining Corollary \ref{C_MS} and \ref{C_EL} we get,

\begin{align*}
U_1(t,t) - U_1(a,t) =& \sum_{i,j \in V}t(i)t(j)u_1(i,j) + r\cdot t(a) \cdot \rho + \\
& + t(a) \cdot r \cdot \frac{k^{x}-1}{k^x} + t(a)^2 \cdot \rho - \Big[r \cdot \frac{k^{x}-1}{k^x} + t(a) \cdot \rho \Big] \\ 
\leq & \Big[ \tau + (\rho - \tau) \frac{d-1}{d} \Big]r^2 + r(1-r) \cdot \rho + \\
& + (1-r)r \frac{k^{x}-1}{k^{x}} + (1-r)^2 \cdot \rho - r \frac{k^{x}-1}{k^x} - (1-r) \cdot \rho \\
= & \Big[ \tau + (\rho - \tau) \frac{d-1}{d} \Big]r^2 - \frac{k^{x}-1}{k^{x}} r^2 \\
= & \frac{r^2}{d} \Big[ \tau - (1-\rho) (d-1) - (1- \frac{d}{k^{x}}) \Big] \\
= & \frac{r^2}{d}E \qquad \text{, where } E=\tau - (1-\rho) (d-1) - (1- \frac{d}{k^{x}}) \;.
\end{align*}

If we can show that $E<0$ then strategy $a$ is an ESS. We show why $E<0$:

Let's define the following function:
\begin{align*}
f(k,d,\rho) = (1-\rho)(d-1) + 1- \frac{d}{k^x}  & \text{ , with the restrictions: } k \geq d+1 , 1 \leq d \leq n , x \geq 3 \;.
\end{align*}

Then we define the function $g(d,\rho)$:
\begin{align*}
g(d,\rho) = \min_{k}f(k,d,r) &= (1-\rho)(d-1) + 1- \frac{d}{(d+1)^x} \;.
\end{align*}

Now, the first two partial derivatives of $g(d,\rho)$ with respect to variable $d$, are:
\
\begin{align*}
&\frac{\partial g(d,\rho)}{\partial d} = (1 - \rho) + \frac{(x-1)d-1}{(d+1)^{x+1}} \qquad & \\
&\frac{\partial^2 g(d,\rho)}{\partial d^2} = \frac{-x[(x-1)d-2]}{(d+1)^{x+2}} \qquad   & \text{, which is non-positive for } d \geq 1, x \geq 3 \;.
\end{align*}
This means that function $g$ has its minimum either for $d=1$ or $d=n$: 
\begin{align*}
& g(1,\rho) = 1 - \frac{1}{2^x} \\
& g(n,\rho) = (1-\rho)(n-1) + 1 - \frac{n}{(n+1)^x}
\end{align*}

If the minimum is $g(1,\rho)$:
\begin{align*}
g(1,\rho) \leq g(n,\rho) \text{, or equivalently, } \rho \leq 1 + \frac{(n+1)^x - n2^x}{2^{x}(n+1)^{x}(n-1)} \;.
\end{align*}

Then,
\begin{align*}
h(\rho) = \min_{d}g(d,\rho) = 1 - \frac{1}{2^x} \;.
\end{align*}

If the minimum is $g(n,\rho)$:
\begin{align*}
g(n,\rho) < g(1,\rho) \text{, or equivalently, } \rho > 1 + \frac{(n+1)^x - n2^x}{2^{x}(n+1)^{x}(n-1)} \;.
\end{align*}

Then,
\begin{align*}
h(\rho) = \min_{d}g(d,\rho) = (1-\rho)(n-1) + 1 - \frac{n}{(n+1)^x} \;.
\end{align*}
\\
So, following the notation we used in Subsection \ref{Extension}:
\[
\tau^* = \min_{k,d}f(k,d,\rho) = 
\begin{cases} 
1 - \frac{1}{2^x}        &  \text{, if }   \rho \leq 1 + \frac{(n+1)^x - n2^x}{2^{x}(n+1)^{x}(n-1)}  \\
\\
(1-\rho)(n-1) + 1 - \frac{n}{(n+1)^x}  &  \text{, if }   \rho > 1 + \frac{(n+1)^x - n2^x}{2^{x}(n+1)^{x}(n-1)}    \\
\end{cases}
\] 
Therefore, we can demand $\tau$ to be strictly less than $\tau^*$, making $U_{1}(t,t) - U_{1}(a,t)$ negative. We conclude that when $d<k$ then strategy $a$ is an ESS.

\subsubsection{Case 2: $d \geq k$}. Let $C \subseteq V$ be a clique of $G$ of size $k$. Then $t$ with $t(i)=\frac{1}{k}$ for $i \in C$ and $t(j)=0$ for $j \in S \setminus C$ is a best response to $a$ and $t \neq a$, and
\begin{align*}
&U_1(t,t) = \sum_{i,j \in C}t(i)t(j)u_1(i,j)= \frac{1}{k^2} \cdot (k-1)k \cdot \rho + \frac{1}{k^2}k \cdot \tau = \frac{(k-1) \rho + \tau}{k} \;, \\
&U_1(a,t) = \frac{k^{x}-1}{k^{x}} \;.
\end{align*}

Then,
\begin{align*}
U_1(t,t) - U_1(a,t) &= \frac{1}{k} \Big[\tau - (1 - \rho)(k-1) - (1 - \frac{1}{k^{x-1}}) \Big] \\
&= \frac{1}{k}E' \qquad \text{, where } E'=\tau - (1 - \rho)(k-1) - (1 - \frac{1}{k^{x-1}}) \;.
\end{align*}

If $E' \geq 0$ then $a$ cannot be an ESS. We explain why $E' \geq 0$:

Let's define the following function:
\begin{equation*}
y(k,\rho) = (1 - \rho)(k-1) + 1 - \frac{1}{k^{x-1}}  \text{\quad , with the restrictions: } k \leq d \;.
\end{equation*}

Then we define the function $z(d,\rho)$:
\begin{align*}
z(d,\rho) = \max_{k}y(k,\rho) = (1 - \rho)(d-1) + 1 - \frac{1}{d^{x-1}} \;,
\end{align*}

so,
\begin{align*}
\tau^{**} = \max_{d}z(d,\rho) = (1 - \rho)(n-1) + 1 - \frac{1}{n^{x-1}} \;.
\end{align*}

Now, given that $\tau$ needs to be at least $\tau^{**}$ but strictly less than $\tau^{*}$ the following should hold:
\begin{align*}
(1 - \rho)(n-1) + 1 - \frac{1}{n^{x-1}} < 1 - \frac{1}{2^x}  \text{ , or equivalently, } \rho > 1 + \frac{n^{x-1}-2^x}{2^{x}n^{x-1}(n-1)} \;.
\end{align*}

So we conclude that when $d \geq k$ then strategy $a$ is not an ESS. This completes the proof of Lemma \ref{Lemma2} and Theorem \ref{EL_x}.

\end{proof}
\end{proof}

\begin{corollary}
The \textsc{ess} problem with payoff values in the domains given in Theorem \ref{EL_x} is \textbf{coNP}-hard.
\end{corollary}

\section{Our Main Result}\label{Main result}
Now we can prove our main theorem:

\begin{theorem}{\label{T_MS}}
Any reduction as in Theorem \ref{EL_x} for $x=x_{0} \geq 3$ from the complement of the \textsc{clique} problem to the \textsc{ess} problem is robust under arbitrary perturbations of values in the intervals: 
\begin{align*}
& \tau \in \left[1 - \frac{1}{2^{x_0}} - D, 1 - \frac{1}{2^{x_0}} - D + B \right), \\
& \rho \in \left(1 + \frac{(n+1)^{x_0} - n2^{x_0}}{2^{x_0}(n+1)^{x_0}(n-1)}, 1 + \frac{(n+1)^{x_0} - n2^{x_0}}{2^{x_0}(n+1)^{x_0}(n-1)} + A \right), \\
& \lambda \in \left[1-\frac{1}{k^{x_0}}, 1-\frac{1}{k^{x_1}} \right],
\end{align*}
where $x_{1} \in \left(x_0 , x_{0}\log_{n}(n+1)\right)$, $C = \frac{(n+1)^{x_0}-n^{x_1}}{n^{x_1-1}(n+1)^{x_0}(n-1)}$, $D = C(n-1)$, any $A \in (0,C)$ and $B = (C-A)(n-1)$.
\end{theorem}

\begin{proof}
We denote three partitions of the game's payoff matrix $U$: $U_\tau, U_\rho, U_\lambda$ disjoint sets, with $U_\tau \cup U_\rho \cup U_\lambda = U$ and values $\tau,\rho,\lambda$ of their entries respectively. Each set's entries have the same value. For every $\lambda \in \left[ 1-\frac{1}{k^{x_0}}, 1-\frac{1}{k^{x_1}} \right]$ there is a $x=-\log_k(1-\lambda)$ in the interval $[x_0, x_1]$ such that $\lambda = 1-\frac{1}{k^{x}}$, where $x_0 \geq 3$ and $x_1 \in (x_0, x_{0}\log_n(n+1))$. We will show that, for this $x$, any reduction with the values of $\tau, \rho$ in the respective intervals stated in Theorem \ref{EL_x}, is valid.

In Figure \ref{fig1}, we show the validity area of $\tau$ depending on $\rho$ with parameter $x$, due to Theorem \ref{EL_x}. The thin and thick plots bound the validity area (shaded) for $x=x_0$ and $x=x_1$ respectively.

\begin{figure}[!ht]
    \centering
    \includegraphics[scale=0.75]{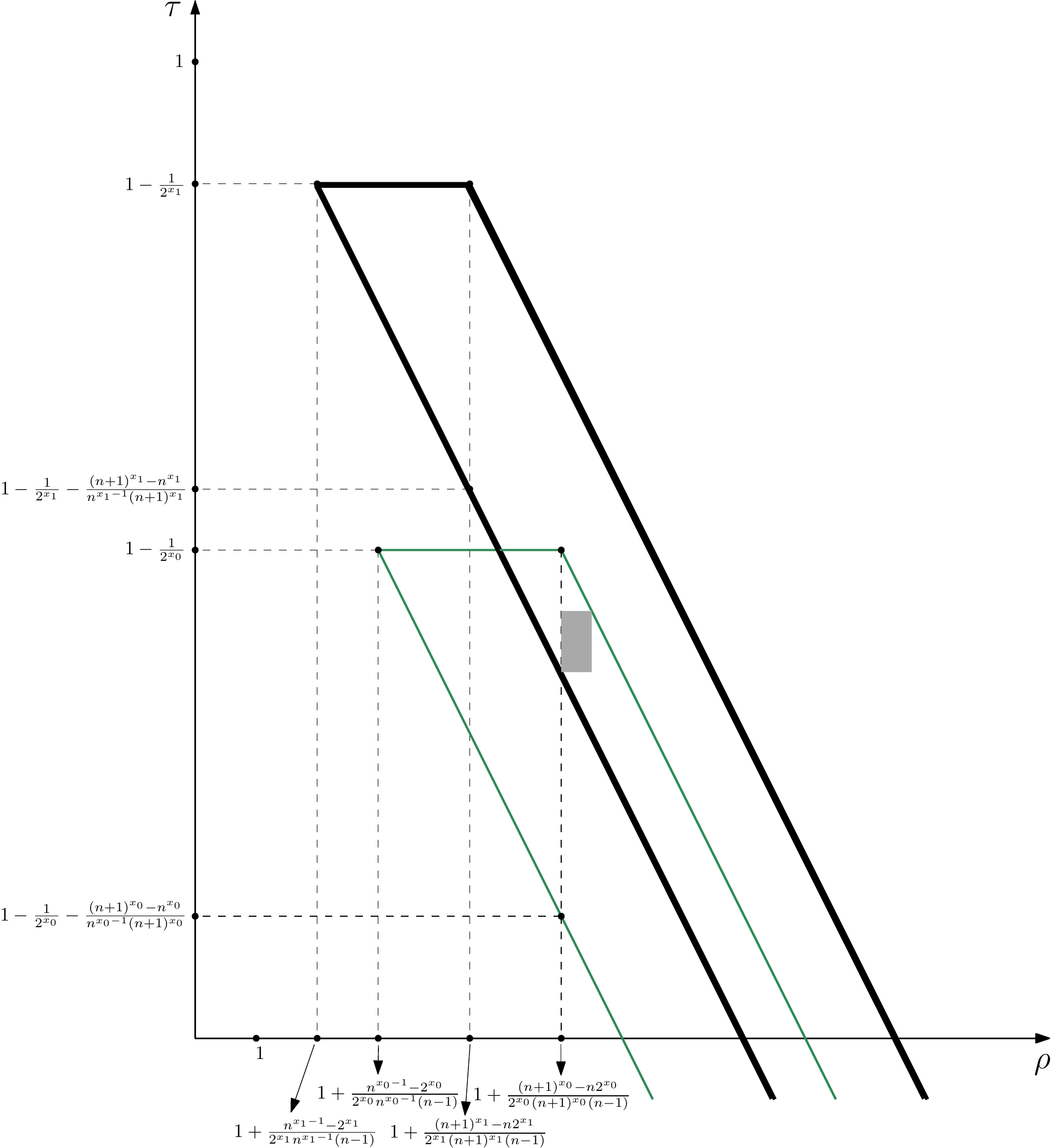}
    \caption{The validity area of $\tau$ and $\rho$ with parameter $x$.}
    \label{fig1}
\end{figure}

While $x$ increases, the parallel lines of the lower and upper bound of $\tau$ move to the right, the horizontal line of the upper bound of $\tau$ moves up, and the left acute angle as well as the top obtuse angle of the plot move to the left (by examination of the monotonicity of those bounds with respect to $x$).

The lower bound of $\tau$ for an $x=x' > x_0$ equals the upper bound of $\tau$ for $x=x_0$, when $x'=x_0\log_n(n+1)$. \textbf{Thus, for all} $\mathbf{x \in (x_0, x_0\log_n(n+1))}$ \textbf{there is a non-empty intersection between the validity areas.} We have picked an $x=x_1 \in (x_0, x_0\log_n(n+1))$.

In Figure \ref{fig2}, we show a zoom-in of the intersection of the validity areas of Figure \ref{fig1}. Let the intersection of lines: $1-\frac{1}{2^{x_0}}$ , $(1-\rho)(n-1) + 1 - \frac{1}{n^{x_1-1}}$ be at point $\rho=\rho_C$.

\begin{figure}[!ht]
    \centering
    \includegraphics[scale=0.75]{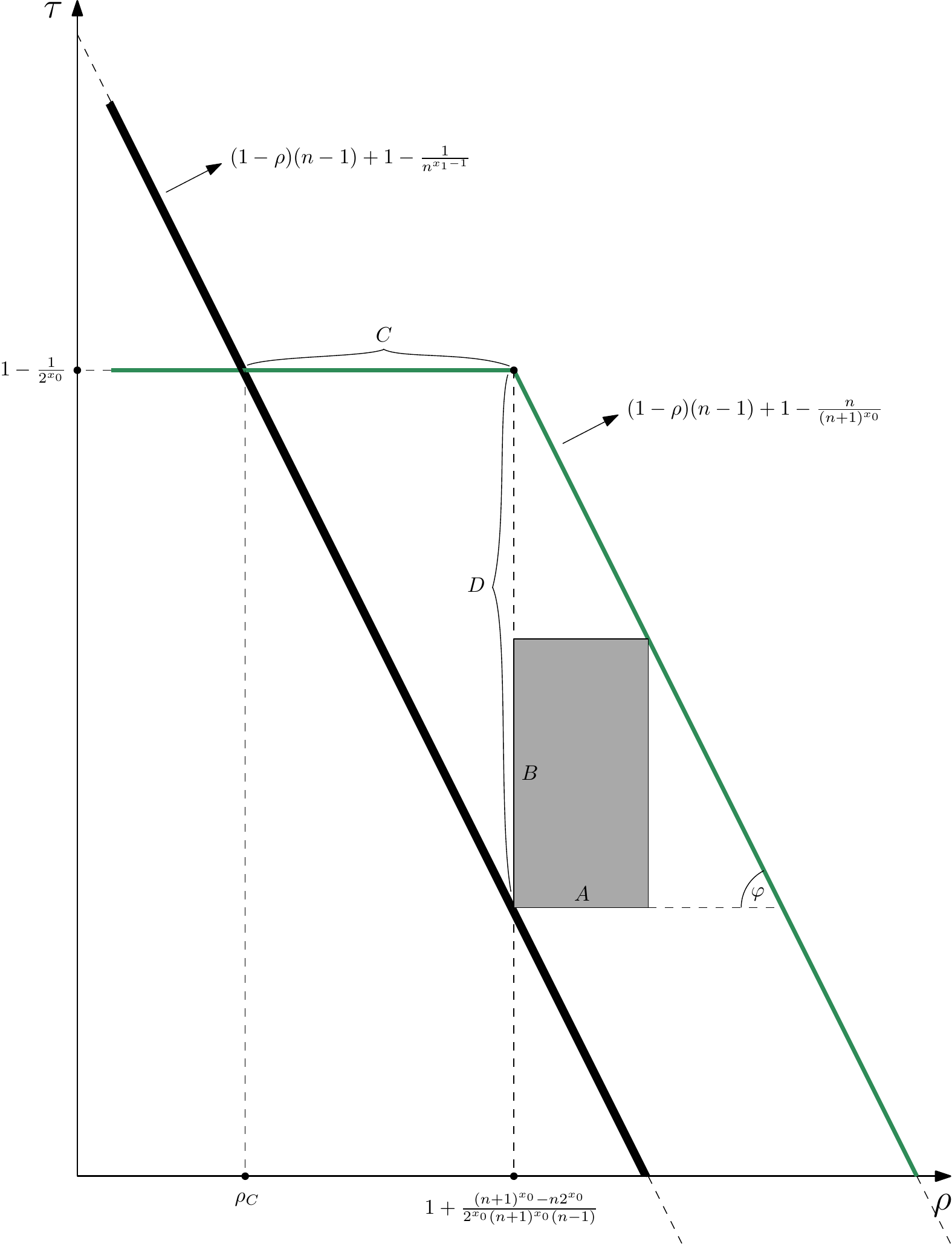}
    \caption{Detail of the validity areas' intersection and the $\rho$, $\tau$ robust area (shaded).}
    \label{fig2}
\end{figure}

 Then, 
\begin{align*}
(1-\rho_C)(n-1) + 1 - \frac{1}{n^{x_1-1}} = 1-\frac{1}{2^{x_0}} \\ 
\text{or equivalently, } \rho_C = 1 - \frac{1}{2^{x_0}(n-1)} - \frac{1}{n^{x_1-1}(n-1)} \;.
\end{align*}

So, 
\begin{align*}
C = 1 + \frac{(n+1)^{x_0} - n2^{x_0}}{2^{x_0}(n+1)^{x_0}(n-1)} - \rho_C \text{, or equivalently, } C = \frac{(n+1)^{x_0}-n^{x_1}}{n^{x_1-1}(n+1)^{x_0}(n-1)} .
\end{align*}

From the upper bound of $\tau$ as a function of $\rho$ we can see that $\tan \varphi = n-1$. Thus,
\begin{align*}
D = C \tan \varphi \text{, or equivalently, } D = \frac{(n+1)^{x_0}-n^{x_1}}{n^{x_1-1}(n+1)^{x_0}} .
\end{align*}

Now we can pick any $A \in (0,C)$. So, it must be 
\begin{align*}
B = (C-A) \tan \varphi \text{, or equivalently, } B = (n-1)(C-A). 
\end{align*}

For the rectangle with sides $A,B$ shown in Figure \ref{fig2}, the reduction is valid for all $x \in [x_0,x_1]$, thus for all $\lambda \in \left[1-\frac{1}{k^{x_0}}, 1-\frac{1}{k^{x_1}} \right]$. This completes the proof of Theorem \ref{T_MS}.
\end{proof}

\section{Conclusions and Further Work}\label{Conclusions}

In this work we introduce the notion of reduction robustness under arbitrary perturbations within an interval and we provide a generalized reduction based on the one in \cite{EL} that proves \textbf{coNP}-hardness of \textsc{ess}. We demonstrate that our generalised reduction is robust, thus showing that the hardness of the problem is preserved even after certain arbitrary perturbations of the payoff values of the derived game. As a future work we would like to examine the robustness of reductions for other hard problems, especially game-theoretic ones.

\newpage

%
%
\bibliographystyle{abbrv}
\bibliography{references}

\end{document}